\documentclass{article}
\usepackage[utf8]{inputenc}
\usepackage[english]{babel}
\usepackage[utf8]{inputenc}

\usepackage{lscape}
\usepackage{amssymb}
\usepackage{amsmath}
\usepackage{mathtools}
\usepackage{stmaryrd}
\usepackage[hypcap=false]{caption}
\usepackage{multirow}

\usepackage{amsthm}
\usepackage{makeidx}
\usepackage[color,all]{xy}
\usepackage{xcolor,colortbl}

\usepackage[hmargin={3cm,3cm},vmargin={2cm,2cm}]{geometry}

\usepackage[pdftex]{graphicx}
\usepackage{comment}
\usepackage{tikz}

\usepackage[colorlinks=true,linkcolor=black]{hyperref}

\usepackage{color}
\usepackage{tikz}
\usetikzlibrary{positioning}
\usetikzlibrary{matrix}

\hypersetup{urlcolor=blue, citecolor=red}

\newcounter{theorem}
\newtheorem{Theorem}{Theorem}[section]

\newtheorem{lemma}[theorem]{Lemma}

\newtheorem{proposition}{Proposition}

\theoremstyle{definition} 

\newtheorem{definition}[theorem]{Definition}
\newtheorem{remark}{Remark}

\newcommand{\be}{\begin{equation}}
\newcommand{\ee}{\end{equation}}
\newcommand{\bea}{\begin{eqnarray}}
\newcommand{\eea}{\end{eqnarray}}
\newcommand{\bd}{\begin{displaymath}}
\newcommand{\ed}{\end{displaymath}}

\begin{document}

\centerline{\Large \bf On   $q^2$-trigonometric functions   and their   $q^2$-Fourier   transform}

\vskip 0.20cm

\centerline{Sama Arjika}

\vskip 0.20cm

\centerline{${}^1$Department of Mathematics and Informatics,   University of Agadez}
\centerline{ and }
\centerline{${}^2$International Chair of Mathematical Physics and Applications, (ICMPA-UNESCO Chair)}
\centerline{ University of Abomey-Calavi, 072 B. P.: 50 Cotonou, Benin}

\begin{abstract}
In this paper, we first construct    generalized $q^2$-cosine,  $q^2$-sine and $q^2$-exponential functions.   We then use   $q^2$-exponential function in order to define and investigate a $q^2$-Fourier transform.  We establish  $q$-analogues of  inversion   and   Plancherel theorems.
\end{abstract}

\noindent


\section{Introduction}

During the last years, an intensive work was founded about the so-called $q$-basic theory.  Taking account of the well-known Ramanujan works shown at the beginning of this century by Jackson [1,2], many authors such as Askey, Gasper,  Rogers, Andrews, Koornwinder, Ismail, Srivastava, and others (see references) have recently developed this topis. 

The present article is devoted  to extend the study of the  $q^2$-analogue of the Fourier transforms. The method used here differs from those given by  Richard  [3].  We take as definition a general form of $q^2$-cosine  [4] 
\bea
\label{qalpha}
\cos(x;q^2) =\sum_{k=0}^{\infty} \frac{(-1)^k\;q^{k(k+1) }\,x^{2k}}{(q;q)_{2k}} 
  =\frac{(q^2;q^2)_\infty}{(q;q^2 )_\infty}\,x^{1/2}\,J_{-1/2}(x;q^2)
\eea
 and $q^2$-sine  [4]
\bea
\label{qoalpha}
 \sin(x;q^2)= \sum_{k=0}^{\infty} \frac{(-1)^k\;q^{k(k+1) }\; x^{2k+1}}{(q;q)_{2k+1}} 
 = \frac{(q^2;q^2)_\infty}{(q;q^2 )_\infty}\,x^{1/2}\,J_{1/2 }(x;q^2),
\eea
where  $J_\alpha(x;q^2)$ the  $q^2$-analogue of the Bessel function  [5,6] 
\bea
\label{Lag}
 J_{\nu }(x;q)=\frac{(q^{\nu +1};q )_\infty}{(q;q)_\infty}x^\nu
 {}_1\phi_{1}\left(\begin{array}{c}0\\
q^{ \nu+1}\end{array}\Big|\;q;\;q^2x^2\right).
\eea

  The $q^2$-analogue Bessel functions and closely related variants have received 
much attention because of their importance in the study of
 $q$-analogues of  representations of the Group of Plane Motions and of the Quantum
Group of Plane Motions, q-differential equations, and other topics.  For more details, see [7-11]. 

Our   aim in this paper is to give an extension of  $q^2$-analogue trigonometric functions  $\cos(x;q^2),\,\sin(x;q^2)$ and $q^2$-analogue exponential   function  $e(x;q^2)$  [4]. We then  study   generalized $q^2$-Fourier transform and give the $q$-analogues of  inversion   and   Plancherel theorems.  

The paper is organized as follows: In Section 2, we 
 give notations and preliminaries to be used in the sequel. In  
  Section 3, we define generalized $q^2$-cosine, $q^2$-sine and $q^2$-exponential functions and study some of their properties. We  give   $q$-analogues of  inversion   and   Plancherel theorems.   We end with concluding remarks in Section 4.
\section{ Notations and preliminaries}

Throughout this paper, we assume that $0<q<1,\;\alpha>-1$  and we write  $\mathbb{R}_{q,+}=\{ q^{n}, n \in \mathbb{Z}\}$. 
We  follow the notations and terminology in  [12-14]. The basic hypergeometric series $ {}_r\phi_s$
\bea
  {}_r\phi_s\left(\begin{array}{c}
a_1, a_2, \cdots, a_r\\
b_1, b_2, \cdots, b_s \end{array}\Big| q;x\right)=  
\sum_{k=0}^\infty\left[(-1)^kq^{k(k-1)/2}\right]^{1+s-r} \frac{(a_1, a_2, \cdots, a_r;q)_k}{(b_1, b_2, \cdots, b_s;q)_k}\frac{x^k}{(q;q)_k},
\eea
converges absolutely for all $x$ if $r\leq s$ and for $|x|<1$ if $r = s + 1$ and for terminating. The compact factorials of $ {}_r\phi_s$ are defined, respectively, by:
\bea
 (a;q)_{0}=1,\  (a;q)_{n}=\prod_{k=0}^{n-1}(1-aq^{k}), 
   (a;q)_{\infty}=\prod_{k=0}^{\infty}(1-aq^{k})
\eea
and
\bea
  (a_1, a_2, \cdots, a_m;q)_n 
  =(a_1;q)_n(a_2;q)_n\cdots(a_m;q)_n, 
\eea
 where $m\in\mathbb{N}:=\{1, 2, \cdots\}$ and $n\in\mathbb{N}_0=\mathbb{N}\cup\{0\}.$

For a complex number $x$ and $n\in \mathbb{N}$,
the $q$-numbers and the $q$-factorials are defined as follows:
\begin {equation*}
[ x]_{q}={{1-q^x}\over{1-q}}, ~   [n]_q! =\prod_{k=1}^n[k]_q,   ~[0]_{q}=1.
\end {equation*}
For $\alpha>-1$, we define the generalized $q$-integers by [15] 
\bea\label{nfactalpha}
[2n]_{q,\alpha}=[2n+2\alpha+1]_{q}, 
[2n+1]_{q,\alpha}=[2n+2\alpha+2]_{q}
\eea
and the generalized $q$-shifted factorials by
\be
\label{qqnalpha}
(q;q)_{n,\alpha}:=(1-q)^n[n]_{q,\alpha}!.
\ee
Remark that,  we can  rewrite (\ref{qqnalpha}) as
$$ 
(q;q)_{2n,\alpha}=(q^2;q^2)_{n}(q^{2\alpha+2};q^2)_{n}
$$
and
$$
(q;q)_{2n+1,\alpha}=(q^2;q^2)_{n}(q^{2\alpha+2};q^2)_{n+1}.
$$
By means of (\ref{nfactalpha}), we may express the generalized $q$-factorials  as
\bea 
[2n]_{q,\alpha}!= \frac{\Gamma_{q^2}(\alpha+n+1)\Gamma_{q^2}(n+1)}{(1+q)^{-2n}\Gamma_{q^2}(\alpha+1)} \eea
and
\bea
[2n+1]_{q,\alpha}!= \frac{(1+q)^{2n+1}\Gamma_{q^2}(\alpha+n+2)\Gamma_{q^2}(n+1)}{\Gamma_{q^2}(\alpha+1)}, 
\eea 
where $\Gamma _q$ is the $q$-Gamma function  given by [12]
\bea 
\Gamma _q(z)=\frac {(q;q)_\infty }{(q^z;q)_\infty }(1-q)^{1-z},  \nonumber
\eea
and tends to $\Gamma (z) $ when $q$ tends to $1^{-}$.
In particular, we have the limits
\bea \label{limitefac}
 \lim_{q\rightarrow1^-}[2n]_{q,\alpha}!=  \frac{2^{2n}n!\Gamma(\alpha+n+1)}{\Gamma(\alpha+1)}=  2^{2n}n! (\alpha+1)_n,  \eea
and
\bea
 \lim_{q\rightarrow1^-}[2n+1]_{q,\alpha}!=\frac{n!\Gamma(\alpha+n+2)}{2^{-2n-1}\Gamma(\alpha+1)} 
= 2^{2n+1}n! (\alpha +1)_{n+1}, 
\eea
where $(a)_k=a(a+1)\cdots(a+k-1)$ is  the Pochhammer-symbol [12,16]. 
\\
Remark that, for  $\alpha=-\frac{1}{2}$, we get:
\be
(q;q)_{n,-\frac{1}{2}}= (q;q)_{n} \;[n]_{q,-\frac{1}{2}}!=[n]_{q}!.  
\ee
The $q$-Jackson integrals    from  $0$ to $+\infty$ and  from $-\infty$ to $+\infty$  are defined by [1]
\be
\label{eq:integral}
\int_{0}^{\infty}{f(x)d_qx}
=(1-q)\sum_{n=-\infty}^{\infty}q^n f(q^n)
\ee
and
\bea
\label{eqintegral}
 \int_{-\infty}^{\infty}{f(x)d_qx} 
 =(1-q)\sum_{n=-\infty}^{\infty}q^n[ f(q^n)+
 f(-q^n)], 
\eea
provided the sums converge absolutely.
\section{Results and Discussion}
In this section, we define and study   generalized $q^2$-trigonometric functions.  We then introduce a $q$-Fourier transform  that formally tends to its classical analogue as $\alpha=-1/2$, $q\to1^{-}$ and study some of its properties.  

\subsection{Generalized $q^2$-analogue trigonometric  functions}

We recall that the $q^2$-analogue exponential function $e(x;q^2)$ is defined in [3]     by
\bea
\label{xpoalpha}
e(x;q^2) =   \cos(-ix;q^2) 
    + i\sin(-ix;q^2).
\eea
By means of   generalized $q$-shifted factorials $(q;q)_{2n,\alpha}$ and $(q;q)_{2n+1,\alpha}$,   we define     generalized  $q^2$-cosine and $q^2$-sine    as follows:
\begin{definition} 
   For $x\in\mathbb{C}$ and $\alpha>-1$,  the generalized $q^2$-cosine and $q^2$-sine  are defined by:
\bea  
\cos_{\alpha}(x;q^2):=  \sum_{k=0}^{\infty} (-1)^k  c_{k,\alpha}(x;q^2)
={}_1\phi_{1}\left(\begin{array}{c}0\\
q^{2\alpha+2}\end{array}\Big|\;q^2; \;q^2x^2\right) \label{q-Expo-alpha} 
\eea 
and
\bea
\sin_{\alpha}(x;q^2):= \sum_{k=0}^{\infty}(-1)^k  s_{k,\alpha}(x;q^2) \label{q-expoalpha}
=\frac{x}{1-q^{2\alpha+2}}\,{}_1\phi_{1}\left(\begin{array}{c}0\\
q^{2\alpha+4}\end{array}\Big|\;q^2; \;q^2x^2\right)  
\eea
where we have put
$$
c_{k,\alpha}(x;q^2) =  \frac{ q^{k(k+1) }\,x^{2k}}{(q^{2\alpha+2},q^2;q^2)_{k}}
$$
and
$$
s_{k,\alpha}(x;q^2) = \frac{ q^{k(k+1) }\,x^{2k+1}}{(q^{2\alpha+2};q^2)_{k+1}(q^2;q^2)_k}.
$$
\end{definition}
By means of   (\ref{q-Expo-alpha}) and (\ref{q-expoalpha}), we define  generalized $q^2$-analogue exponential function $e_{\alpha}(x;q^2)$ by
\bea 
\label{expoalpha}
e_{\alpha}(x;q^2):= \cos_{\alpha}(-ix;q^2) 
+i\sin_{\alpha}(-ix;q^2).
\eea 
\begin{remark}
$\,$
\begin{enumerate}
\item   Compared with $ \cos(x;q^2),\,\sin(x;q^2)$ and $e(x;q^2)$, the generalized $q^2$-cosine and $q^2$-sine and exponential functions   $\cos_{\alpha}(x;q^2),\,\sin_{\alpha}(x;q^2)$ and  $e_{\alpha}(x;q^2)$, respectively,   involve two   parameters  ``$q$" and ``$\alpha$". Clearly,  $\cos(x;q^2),\,\sin(x;q^2)$ and  $e(x;q^2)$   can be considered as a special case  of (\ref{q-Expo-alpha}), (\ref{q-expoalpha})  and (\ref{expoalpha}), respectively. For  $\alpha =-1/2,$ we have:
\bea
\cos_{-\frac{1}{2}}(x;q^2)=\cos(x;q^2), 
 \sin_{-\frac{1}{2}}(x;q^2)=\sin(x;q^2), 
 e_{-\frac{1}{2}}(x;q^2)=e(x;q^2). 
\eea
\item  The relation between   generalized $q^2$-cosine and $q^2$-sine    functions 
 and the classical hypergeometric functions is ba-sed on observations such as
\bea
\lim_{q\to 1^{-}} \cos_{\alpha}((1-q)x;q^2) 
={}_0F_{1}\left(\begin{array}{c}-\\
 \alpha+1 \end{array}\Big|-\frac{x^2}{4}\right) 
\eea
and
\bea
 \lim_{q\to 1^{-}} \sin_\alpha((1-q)x;q^2) = \frac{x}{2(\alpha+1)}\,{}_0F_{1}\left(\begin{array}{c}-\\
 \alpha+ 2\end{array}\Big|-\frac{x^2}{4}\right). 
\eea
\item  For $\alpha=-1/2$, we have: 
\bea
  \lim_{q\to 1^{-}} \cos_{-\frac{1}{2}}((1-q)x;q^2)=\cos x, 
 \lim_{q\to 1^{-}} \sin_{-\frac{1}{2}}((1-q)x;q^2)=\sin x, 
 \lim_{q\to 1^{-}}e_{-\frac{1}{2}}((1-q)x;q^2)=e^x.\nonumber
\eea
\end{enumerate}
The   generalized $q^2$-exponential function $e_\alpha(x;q^2)$ is   absolutely convergent for all $x$  in the plane, $0 < q <1$, since both generalized $q^2$-cosine and $q^2$-sine are absolutely convergent for all $x$  in the plane, $0 < q <1$. 
\end{remark}

We introduce   generalized $q$-differential  operator as
\bea\label{D_qalpha}
 \partial_{q,\alpha}f(x)&=  & \frac{f(q^{-1}x)+f(-q^{-1}x)-q^{2\alpha+1}[f( x)+f(-x)]  }{2(1-q)x} \cr
& +&  \,  \frac{f(x)-f(-x) - q^{2\alpha+1}[f(q x)-f(-q x)]   }{2(1-q)x},   \,  x\neq0 
\eea
and 
$$ 
 \partial_{q,\alpha}f(0)=\lim_{x\to 0}\partial_{q,\alpha}f(x)=[2\alpha+2]_qf'(0), $$ provided that $ f'(0)$ exists.\\
We notice   if f is differentiable at $x$   $$ \lim_{q\to 1^-}\partial_{q,\alpha}f(x)=f'(x).$$
Observe that, $\alpha=-1/2$ in (\ref{D_qalpha}) corresponds to  the  $ q^2$-analogue differential operator [3], i.e.,    $\displaystyle   \partial_{q,-1/2}f(x)=  \partial_{q}f(x)$, 
and  
$\displaystyle 
 \partial_{q,-1/2}f(0)=f'(0).$\\
For all function $f$ on $\mathbb{R}_{q,+}$,  we have:
\bea
\label{dpalpha}
 \partial_{q,\alpha}f(x)= \frac{f_e(q^{-1}x)  -q^{2\alpha+1}f_e( x) }{(1-q)x} 
 + \,\frac{f_o(x)-q^{2\alpha+1}f_o(qx)   }{(1-q)x}, 
\eea
where  $ f_e$ and $f_o$ are  respectively,  even and odd parts of f.
Since we have a realization of the generalized $q$-differential operator $ \partial_{q,\alpha}$ (\ref{dpalpha}),  we have:
\bea \label{aser}
  \partial_{q,\alpha}\cos_\alpha((1-q)xt;q^2) =- t\sin_\alpha((1-q)xt;q^2),
\eea
\bea 
\label{asert}
 \partial_{q,\alpha}\sin_\alpha((1-q)xt;q^2)=t\cos_\alpha((1-q)xt;q^2)
\eea
and
 \bea 
\label{asertt}
 \partial_{q,\alpha}\,e_\alpha((1-q)xt;q^2)=te_\alpha((1-q)xt;q^2).
\eea
\subsection{Generalized $q^2$- Fourier transform}
The  goal is now to define a generalized $q^2$-deformed Fourier transform  that formally tends to its classical analogue as $\alpha=-1/2$ and $q\to1^{-}.$ 

  For $1\leq p <\infty$, we denote by $\displaystyle  L_{\alpha ,q}^p(\mathbb{R}_{q,+})$ the space of complex-valued functions $f$ on $\mathbb{R}_{q,+}$ such that:
\bea
\|f\|_{q,\alpha,p}=
\left(\int_{-\infty}^{\infty}|f(x)|^p|x|^{2\alpha +1}d_qx\right)^{\frac{1}{p}}<\infty 
\eea
and for $p=\infty$, we denote by $\displaystyle L_{\alpha ,q}^\infty(\mathbb{R}_{q,+})$  the space of complex-valued
functions $f$ on $\mathbb{R}_{q,+}$  such that
$$
\|f\|_{q,\alpha,\infty}= \sup_{x\in \mathbb{R}_{q,+}}\{|f(x)||x|^{2\alpha+1}\}<\infty.
$$
The generalized $q^2$-Fourier transform  will now be defined.
\begin{definition}
Let $f$ be a function in the space  
$\displaystyle L_{\alpha,q}^1(\mathbb{R}_{q,+}).$  The generalized $q^2$-Fourier transform is defined  by:
\bea
\hat{f}(x;q^2):= C_{\alpha,q}  \int_{-\infty}^{\infty} f(t)\;e_{\alpha}(-i(1-q)tx;q^2)
  |t|^{2\alpha+1}\,d_qt, 
\eea
where $\displaystyle   C_{\alpha,q}=\frac{(1-q)^\alpha\left(q^{2\alpha+2};q^2\right)_\infty}{2 \,\left(q^{2};q^2\right)_\infty}.$
 \end{definition}

For $\alpha=-1/2$ and letting $q\uparrow 1$ subject to the condition
\be
\label{q}
\frac{Log(1-q)}{Log(q)}\in 2\mathbb{Z},
\ee
gives, at least formally, the classical Fourier transform.  In the remainder of this paper, we assume that the condition (\ref{q})
holds.

The following Lemma will be used to prove the inversion Theorem.  
\begin{lemma}
\label{Lemma 3.1.}
Let $f\in L_{\alpha,q}^1(d_qt)$. If $\displaystyle \int_{-\infty}^{\infty}f(x)|x|^{2\alpha +1}  d_qx$ exists, with $(-1)^{2\alpha+1}=1,$  then:\\
$\bullet$ $f$ odd implies  $\displaystyle \int_{-\infty}^{\infty}f(x)|x|^{2\alpha +1} d_qx=0;$\\
$\bullet$ $f$ even  implies  $\displaystyle \int_{-\infty}^{\infty}f(x)\,|x|^{2\alpha +1} d_qx=2  \int_{0}^{\infty}f(x) \,x^{2\alpha +1}d_qx$.
\end{lemma}
\begin{proposition}
\label{Proposition 3.1.}
For $f,g\in  L_{\alpha, q}^2(\mathbb{R}_{q,+}),$ the generalized $q^2$-cosine and $q^2$-sine transforms pair  hold true.
\bea 
\label{seds}
g(q^n)=2  C_{\alpha,q} \int_{0}^{\infty}\left\{\begin{array}{c}
  \cos_\alpha((1-q)xq^n;q^2)\\
\mbox{ or }\\
 \sin_\alpha((1-q)xq^n;q^2)\end{array}\right\}  f(x)x^{2\alpha+1}\,d_qx   
\eea
and 
\bea  
\label{sedss}
f(q^k)=2  C_{\alpha,q}   \int_{0}^{\infty} \left\{\begin{array}{c}
  \cos_\alpha((1-q)tq^k;q^2)\\
\mbox{ or }\\
 \sin_\alpha((1-q)tq^k;q^2)\end{array}\right\} g(t)t^{2\alpha+1}\,d_qt.   
\eea
\end{proposition}
\begin{proof} 
In order to prove   the Proposition \ref{Proposition 3.1.}, we will start with the relation [4]
\bea
\label{labelm}
\delta_{nm}=\sum_{k=-\infty}^\infty  z^{k+n}\frac{(z^2;q)_\infty}{(q;q)_\infty}  {}_1\phi_1\left(\begin{array}{c}0\\
z^2\end{array}\Big|q; q^{n+k+1}\right)  z^{k+m}\frac{(z^2;q)_\infty}{(q;q)_\infty}{ }_1\phi_1\left(\begin{array}{c}0\\
z^2\end{array}\Big|q; q^{m+k+1}\right), 
\eea
where   $|z|<1, \,n, \,m\in\mathbb{Z}$. \\
Substituting  $q$ by $q^2$ and  $z$ by  $q^{\alpha+ 1}$ into (\ref{labelm}) yields 
\bea
\label{resd}
\delta_{nm}=\sum_{k=-\infty}^\infty  
   q^{( \alpha+ 1)(k+n)}\frac{(q^{2\alpha+ 2};q^2)_\infty}{(q^2;q^2)_\infty}\cos_\alpha(q^{n+k};q^{2})  q^{( \alpha+ 1)(k+m)}\frac{(q^{2\alpha+ 2};q^2)_\infty}{(q^2;q^2)_\infty}\cos_\alpha(q^{m+k};q^{2}).  
\eea
Rewrite the identity (\ref{resd}) as the transform pair  
\bea
 g(q^n)= \sum_{k=-\infty}^\infty q^{(n+k)(\alpha+1 )}\frac{(q^{2\alpha+ 2};q^2)_\infty}{(q^2;q^2)_\infty} \cos_\alpha(q^{n+k};q^{2}) f(q^k)  \label{sresdd}
\eea 
and
\bea 
f(q^k)= \sum_{n=-\infty}^\infty q^{(n+k)(\alpha+1 )}\frac{(q^{2\alpha+ 2};q^2)_\infty}{(q^2;q^2)_\infty}  \cos_\alpha(q^{n+k};q^{2}) g(q^n),  \label{srfesdd} 
\eea 
where $f$ and  $g$   are $L_{\alpha,q}^2$ on the set $\{ q^{k}, 
 k\in\,\mathbb{Z}\}$ with respect to counting measure. 
Replacing  in   (\ref{sresdd})  $f( q^{k}), \,g(q^{n})$ by 
$q^{k(\alpha+1)}f( q^{k})$, $q^{n(\alpha+1)}g(q^{n}),$ respectively, we obtain:
\bea 
 g(q^n)= \sum_{k=-\infty}^\infty q^{k(2\alpha+2 )}\frac{(q^{2\alpha+ 2};q^2)_\infty}{(q^2;q^2)_\infty}  \cos_\alpha(q^{n+k};q^{2}) f(q^k).\label{resded}
\eea 
For such $q\in\,\{ q^{k},\; k\in\,\mathbb{Z}\}$, we can  replace $q^{k}, \,q^{n}$ in (\ref{resded}), by $(1-q)^{1/2}q^{k}$,
$ (1-q)^{1/2}q^{n}$. Then,
\bea 
 g((1-q)^{1/2}q^n)= \frac{(q^{2\alpha+ 2};q^2)_\infty}{(q^2;q^2)_\infty} \sum_{k=-\infty}^\infty    q^{k(2\alpha+2 )} \cos_\alpha((1-q)  q^{n+k};q^{2})  (1-q)^{\alpha+1} f((1-q)^{1/2}q^{k}).  
\eea 
Next,   replacing $f((1-q)^{1/2}q^{k})$ and  
$ g((1-q)^{1/2}q^{n})$  by $f( q^{k})$ and $g(q^{n})$, we get:
\bea 
\label{resdedd}
g(q^n)=(1-q)^{\alpha+1}\frac{(q^{2\alpha+ 2};q^2)_\infty}{(q^2;q^2)_\infty}
  \sum_{k=-\infty}^\infty   q^{ k(2\alpha+2 )}\cos_\alpha((1-q)q^{n+k};q^{2})   f(q^k). 
\eea 
With the $q$-integral notation (\ref{eq:integral}), the relation  (\ref{resdedd}) is equivalent to 
\bea  
g(t)=2   C_{\alpha,q}\int_{0}^{\infty}  \cos_\alpha((1-q)tx;q^2)  f(x)x^{2\alpha+1}\,d_qx.
\eea 
The proof of (\ref{seds}) is achieved.\\
Similarly, we can  prove (\ref{sedss}).  
\end{proof}
 
\begin{remark} 
For  $\alpha=-1/2$  and $q\uparrow 1$ in   assertions (\ref{seds}) and (\ref{sedss}), we get
the classical Fourier pair:
\bea
g(\lambda)=\sqrt{\frac{2}{\pi}} \int_{0}^{\infty} \left\{\begin{array}{c} \cos(x\lambda)\\\mbox{ or }\\
\sin(x \lambda)\end{array}\right\}f(x)\,dx,
\eea
and
\bea
f(x)=\sqrt{\frac{2}{\pi}} \int_{0}^{\infty} \left\{\begin{array}{c} \cos(x\lambda)\\\mbox{ or }\\
\sin(x \lambda)\end{array}\right\}g(\lambda)\,d\lambda. 
\eea
\end{remark}

\begin{lemma} 
For $f,g\in  L_{\alpha, q}^2(\mathbb{R}_{q,+}),$ the transformations $f\mapsto g$ and $g\mapsto f$ of (\ref{seds}) and (\ref{sedss}) establish and isometry of  Hilbert spaces: 
\bea 
\sum_{k=-\infty}^{+\infty}q^{k(2\alpha+2)}|f(q^k)|^2 
=\sum_{n=-\infty}^{+\infty}q^{n(2\alpha+2)}|g(q^n)|^2.
\eea
\end{lemma}
Let us now turn to the $L^2$ theory of the generalized $q^2$-Fourier Transform.    Since the generalized $q^2$-Fourier Transform is defined and bounded on $ \big(L_{\alpha,q}^1\cap L_{\alpha,q}^2\big)(d_qt)$  (dense in  $ L_{\alpha,q}^2(d_qt)$ for the functions with finite support), it defines a bounded extension to all of $ L_{\alpha,q}^2(d_qt)$.  We can use Lemma \ref{Lemma 3.1.} and Proposition \ref{Proposition 3.1.} to prove the following theorem.  
\begin{Theorem}
$f\in \big(L_{\alpha,q}^1\cap L_{\alpha,q}^2\big)(d_qt)$ implies 
\bea
f(x)=  C_{\alpha,q}\int_{-\infty}^{\infty}\hat{f}(t;q^2)\;e_{\alpha}(i(1-q)t x;q^2)
  |t|^{2\alpha+1}\,d_qt,\,\forall x\in\mathbb{R}_{q,+}. 
\eea
\end{Theorem}
\begin{Theorem}
Let $f$ be  the functions with finite support in  $ L_{\alpha,q}^2(d_q x)$. 
$f\in   L_{\alpha,q}^2(d_q x)$ implies  
\be
\|f\|_{q,\alpha,2}=\|\hat{f}(\cdot;q^2)\|_{q,\alpha,2}.
\ee
\end{Theorem}
\section{Conclusions}

In our present investigation, we have constructed a pair of potentially   generalized $q^2$-cosine,  $q^2$-sine and $q^2$-exponential functions.  We then have successfully used $e_\alpha(x;q^2)$  to define and investigate  generalized $q^2$-Fourier transform.  In particular, we have establihed $q$-analogues of   inversion   and  Plancherel theorems. \ 
{}
	
\end{document}